\documentclass[11pt]{article}
\usepackage{amsmath,amsthm,amsfonts,amssymb,amscd,cite,graphicx}
\usepackage{titlesec}
\usepackage{tikz}
\usepackage{color}
\usepackage{cite}
\usepackage{comment}
\usepackage{url}

\newcommand{\Sminus}{S^{-}}
\newcommand{\Away}{{\tt A}}
\newcommand{\Home}{{\tt H}}
\newcommand{\UnHA}{{\tt *}}
\newcommand{\HA}{1}
\newcommand{\AH}{2}
\newcommand{\OCTmap}{\alpha}
\newcommand{\myalpha}[1]{#1}
\newcommand{\bV}{\hat{V}}
\newcommand{\bE}{\hat{E}}
\newcommand{\minB}{B_{\min}}
\newcommand{\minO}{O_{\min}}
\newtheorem{THEO}{Theorem}[section]
\newtheorem{LEMM}[THEO]{Lemma}
\newtheorem{CORO}[THEO]{Corollary}

\titleformat{\section}
{\normalfont\fontsize{12}{15}\bfseries}{\thesection}{1em.}{}

\allowdisplaybreaks[4]

\let\oldbibliography\thebibliography
\renewcommand{\thebibliography}[1]{%
  \oldbibliography{#1}%
  \setlength{\itemsep}{-2pt}%
}

\begin{document}

\title{Minimizing breaks by minimizing odd cycle transversals
\thanks{
This work was supported by JSPS KAKENHI Grant Number JP20K04973. 
}
}

\author{Koichi Fujii
\thanks{Department of Industrial Engineering and Economics, 
   Tokyo Institute of Technology 
}
       \and
        Tomomi Matsui
\footnotemark[2]
}

\maketitle

\begin{abstract}
 \noindent
 Constructing a suitable schedule for sports competitions is
  a crucial issue in sports scheduling.
The round-robin tournament is a competition adopted 
  in many professional sports.
For most round-robin tournaments,
  it is considered undesirable that
  a team plays consecutive away or home matches;
  such an occurrence is called a break.
Accordingly,
  it is preferable to reduce the number of breaks in a tournament.
A common approach is first to construct a schedule
  and then determine a home-away assignment
  based on the given schedule
  to minimize the number of breaks (first-schedule-then-break).
In this study, 
  we concentrate on the problem
  that arises in the second stage
  of the first-schedule-then-break approach, namely,
  the break minimization problem~(BMP).
We prove that
  this problem can be reduced to an odd cycle transversal problem,
  the well-studied graph problem.
These results lead to a new approximation algorithm for the BMP.  

 \end{abstract}

\section{Introduction}

Constructing a suitable schedule for sports competitions is
  a crucial issue in sports scheduling.
A  (single)  {\it round-robin tournament} (RRT) is a simple sports competition,
  where each team plays against every other team once. 
The RRT is a competition framework adopted 
  in many professional sports
  such as soccer and basketball,
  especially in Europe.
In this study, 
  we consider an RRT schedule with the following properties:

\begin{itemize}
\item Each team plays one match in each \textit{slot},
   that is the day when a match is held.
\item Each team has a home stadium, and each match is held 
  at the home stadium of one of the two playing teams.
  When a team plays a match at the home of the opponent, 
  we state that the team plays away.
\item Each team plays each other team exactly once at home or away. 
\end{itemize}    

It is considered undesirable that
  a team plays consecutive away or home matches
  to reduce player's burden;
  such an occurrence is called a {\it break}.

Accordingly,  
  the number of breaks in a tournament should be reduced.
To construct a tournament schedule with fewer breaks,
  the following two decomposition approaches are widely used:
  (i) first-break-then-schedule~\cite{
    nemhauser1998scheduling,
    briskorn2008feasibility,
    miyashiro2002characterizing,
    henz2001scheduling,
    van2022optimizing,
    zeng2012separation,
    van2020complexity}
  and (ii) first-schedule-then-break~\cite{
    trick2001schedule,
    nemhauser1998scheduling,
    miyashiro2005polynomial,
    miyashiro2006semidefinite,
    post2006sports,
    rasmussen2006timetable}.
In the first-break-then-schedule approach, 
    an home-away assignment~(HA-assignment) is generated in the first stage,  
    which assigns a home game or an away game to each team
    in each slot.
In the second stage,
  a timetable consistent with the generated HA-assignment
  is found,
  if it exists.
Contrarily, in the first-schedule-then-break approach,
  a timetable is constructed at the first stage.
The second stage determines the home and away teams of each match.
In this approach, 
  the home advantage is further determined in the second stage.
For detailed information,  
  refer to~\cite{Rasmussen_2008}.
In this study,
  we focus on a problem
  occurring in the second stage 
  of the first-schedule-then-break approach.
This problem is commonly called the break minimization problem~(BMP),
  which determines an \mbox{HA-assignment} minimizing the number of breaks.

The research on the BMP,
  introduced by the seminal work of
  de Werra~\cite{de1981scheduling},
  has steadily evolved.
By $B_{\min}(\tau)$, 
  we denote the minimum total number of breaks
  over all possible HA-assignments
  for a given timetable $\tau$ of an RRT with $2n$ teams.
For the lower bound of $B_{\min}(\tau)$,
  de Werra~\cite{de1981scheduling} demonstrated that
  any schedule of a round-robin tournament has
  at least $2n- 2$ breaks.
For the upper bound,
 Miyashiro and Matsui~\cite{miyashiro2005polynomial} 
 proposed an algorithm 
 for finding an HA-assignment satisfying the condition
 that the number of breaks 
 is less than or equal to $n(n-1)$ for a given timetable.
Post and Woeginger~\cite{post2006sports} revised this analysis and 
    improved the upper bound of the number of breaks from  $n(n-1)$
    to $(n-1)^2$ if the number of teams $2n$ is not 
    a multiple of~$4$. 

Elf, J\"{u}nger, and Rinaldi~\cite{ELF2003343} conjectured that
  the problem of finding $B_{\min}(\tau)$ is NP-hard,
  which has not been proven yet to the best of our knowledge.
Miyashiro and Matsui~\cite{miyashiro2006semidefinite} formulated
  BMP as MAX RES CUT.
They also applied an approximation algorithm for MAX RES CUT based on the semi-definite programming relaxation,
  which was proposed by Goemans and Williamson~\cite{goemans1995improved}.

In this paper,
  we consider the odd cycle transversal (OCT) problem,
  which is the problem of deleting vertices to transform a graph into a bipartite graph.

After the breakthrough paper of Reed, Smith, and Vetta~\cite{reed2004finding},
  various parameterized algorithms have been investigated~\cite{
    lokshtanov2009simpler, 
    kawarabayashi2010almost,
    jansen2011polynomial,
    lokshtanov2014faster, 
    iwata2014linear,
    kolay2020faster
  }.
Not only are there theoretical results,
  but there are also computational results
  for both the exact and heuristic approaches~\cite{
    huffner2009algorithm,
    akiba2016branch
    }.
We will explore the relationship
  between the BMP and OCT problems and consequently propose a novel approximation algorithm of the BMP.
    
\section{Preliminaries}

\subsection{Break Minimization Problem}
\label{subsection:BMP}
In this section, we introduce some notations and formal definitions.
We use the following notations and symbols throughout this paper:

\begin{itemize}
  \item $2n$: number of teams, where $n \geq 2$, 
  \item $T = \{1,2, \ldots, 2n\}$: set of teams,
  \item $S = \{1,2, \ldots , 2n-1 \}$: set of slots.
\end{itemize}

A schedule of an RRT 
  is described as a pair of a timetable and
  an HA-assignment, which is defined below. 
In this study,
  we assume that a timetable $\tau$ of an RRT is 
  a matrix whose rows and columns are indexed by $T$ and $S$, respectively.
An element $\tau (t,s)$
  denotes the opponent that plays against team $t$ at slot $s$.
A timetable $\tau$ (of an RRT schedule)
  should satisfy the following conditions:
  (i)~a row of $\tau$ indexed by team $t \in T$ 
  is a permutation of teams in $T\setminus \{t\}$,
  and (ii)~$\tau(\tau(t,s), s) = t \; (\forall (t,s)\in T \times S).$
For any pair of different teams
  $t_1, t_2 \in T,$ 
  $s(\{t_1,t_2\})$ indicates a slot in which 
  teams $t_1$ and $t_2$ have a match.
We denote a subset $\{1,2, \ldots , 2n-2 \} \subset S$ as $\Sminus$.
Table~\ref{table:timetable} presents
  a timetable for an RRT schedule of four teams.

\begin{table}[h]
  \begin{minipage}{.50\textwidth}
    \centering
    \caption[short]{Timetable $\tau$}
    \label{table:timetable}
    \begin{tabular}{|l|l|l|l|}
      \hline Slot & $\mathbf{1}$ & $\mathbf{2}$ & $\mathbf{3}$ \\
      \hline
      \hline team 1 & 2 & 3 & 4 \\
      \hline team 2 & 1 & 4 & 3 \\
      \hline team 3 & 4 & 1 & 2 \\
      \hline team 4 & 3 & 2 & 1 \\
      \hline
    \end{tabular}
  \end{minipage}
  \begin{minipage}{.50\textwidth}
    \centering
    \caption[short]{HA-assignment $\mathbf{Z}$}
    \label{table:ha_assignment}
    \begin{tabular}{|l|l|l|l|}
    \hline Slot & $\mathbf{1}$ & $\mathbf{2}$ & $\mathbf{3}$ \\
    \hline
    \hline team 1 & \Home & \Away & \Home \\
    \hline team 2 & \Away & \Away & \Home \\
    \hline team 3 & \Home & \Home & \Away \\
    \hline team 4 & \Away & \Home & \Away \\
    \hline
    \end{tabular}
  \end{minipage}
\end{table}  
\begin{table}[h]
    \centering
    \caption[short]{Partial HA-assignment}
    \label{table:partial_ha_assignment}
    \begin{tabular}{|l|l|l|l|}
    \hline Slot & $\mathbf{1}$ & $\mathbf{2}$ & $\mathbf{3}$  \\
    \hline
    \hline team 1 & \Home & \Away & \Home  \\
    \hline team 2 & \UnHA & \Away & \Home   \\
    \hline team 3 & \UnHA & \Home & \Away \\
    \hline team 4 & \Away & \Home & \Away \\
    \hline
    \end{tabular}
\end{table}

%


A team is defined to be at {\it home} in slot $s$ if the team plays a match 
  at its home stadium in $s$;
  otherwise, it is {\it away} in $s$.
An {\it HA-assignment}
    is a matrix \mbox{$\mathbf{Z} = ( z_{t,s} )$}
    in which the components consist 
    of $\Home$ and $\Away$,  
    whose rows and columns are indexed by $T$ and $S$, respectively. 
A value of $z_{t,s}$ is equal to $\Home$
  if team $t$ plays a match in slot $s$ at home;
  otherwise, it is equal to $\Away$.
For a given timetable $\tau$,
  $\mathbf{Z}$ is {\em consistent} with $\tau$ 
  when $\{z_{t,s},z_{\tau(t,s),s}\}=\{\Home, \Away\}$ 
 holds for each $(t,s)\in T \times S$.
Table~\ref{table:ha_assignment} 
  presents an example of an HA-assignment
  that is consistent with the timetable
  presented in Table~\ref{table:timetable}.

Given an HA-assignment $\mathbf{Z}$,
  we say that team $t$ has a break at slot $s \in S\setminus \{1\}$
  if $z_{t,s-1}=z_{t,s}$ holds.
For example, in Table~\ref{table:ha_assignment}, 
  team $2$ has a break at slot $2$,
  as there are consecutive away matches $(z_{2,1}=z_{2,2}=\Away)$.
The number of breaks in an HA-assignment
  is defined as the total number of breaks belonging to all teams.

Following the definitions introduced above,
  we are now ready to present a formal definition of the BMP:

\noindent
\textbf{Break Minimization Problem}:
 {\it
 Given a timetable $\tau$, 
 the BMP finds an HA-assignment consistent with $\tau$
 that minimizes the number of breaks.
 }

A {\em partial HA-assignment} 
    is a $T \times S$ matrix 
    $\mathbf{Z} = (z_{t,s})$
    with the components consisting of 
    $\{\Home, \Away, \UnHA \}$.
For a given timetable $\tau$,
  a partial HA-assignment 
  $\mathbf{Z}$ is {\em consistent} with $\tau$ when
\begin{align*}
  (z_{t,s},z_{\tau(t,s),s})\in 
   \{(\Home, \Away), (\Away, \Home), 
     (\Home, \UnHA), (\Away, \UnHA), 
     (\UnHA, \Home), (\UnHA, \Away), 
     (\UnHA, \UnHA) \}
\end{align*}     
 holds for each $(t,s)\in T \times S$.
Table~\ref{table:partial_ha_assignment} 
  presents an example of an HA-assignment
  that is consistent with the timetable
  presented in Table~\ref{table:timetable}.
Suppose that   
  we have a partial HA-assignment $\mathbf{Z}$
    consistent with a given timetable $\tau$.
We can construct an HA-assignment $\tilde{Z}$
  consistent with $\tau$ 
  by assigning $\Home$ or $\Away$ 
  to $\UnHA$-components of $\mathbf{Z}$
  so that $\tilde{z}_{t,s} = z_{t,s}$ 
    for $z_{t,s} \in \{\Home, \Away\}$
  and $(\tilde{z}_{t,s},\tilde{z}_{\tau(t,s),s})\in 
    \{(\Home, \Away), (\Away, \Home)\}$ are satisfied.


\subsection{Odd Cycle Transversal}

Given an undirected graph $G=(V,E)$, 
 a vertex subset $V'\subseteq V$ is called 
 an {\em OCT} if and only if  
 $V'$ has a nonempty intersection 
 with every odd cycle in $G$.
In other words, the induced subgraph 
 $G[V\setminus V']$ is a bipartite graph.
The minimum OCT problem finds
 an OCT 
 with the smallest cardinality.
 
Given an OCT $V'$ of $G=(V,E)$,
 we can construct a map 
  $\OCTmap: V \rightarrow (0,1,2)$
 (where $(0,1,2)$ is an ordered triplet 
 of 3-set $\{0,1,2\}$) 
 satisfying that 
 $\OCTmap^{-1}(0)=V'$ and 
 both $\OCTmap^{-1}(1)$ and $\OCTmap^{-1}(2)$ 
 are independent sets of $G.$
 Here, we note that this construction is not unique.
 Conversely, if we have a map 
  $\OCTmap: V \rightarrow (0,1,2)$
   satisfying the condition that
\begin{equation}\label{OCTmap}
    \mbox{ \hspace{-0.5cm} 
 both $\OCTmap^{-1}(1)$ and $\OCTmap^{-1}(2)$ 
 are independent sets of $G,$}
\end{equation}
then $\OCTmap^{-1}(0)$ becomes an OCT.
We call such a map an {\em OCT-map}.
Throughout this paper, 
 we represent an OCT 
 by an OCT-map \mbox{$\OCTmap: V \rightarrow (0,1,2)$}.
 





\section{Main theorem}

To describe our main theorem,
 we consider an auxiliary graph of 
 a given timetable $\tau$.
It is an undirected graph
$G(\tau)=(\bV, \bE)$ 
 with vertex-set $\bV=T \times \Sminus$ 
and edge-set $\bE$ defined by $\bE = \bE_H \cup \bE_0 \cup \bE_1$ where
\begin{eqnarray*}
  \bE_H &=&  \{\{(t,s-1),(t,s)\} \mid  t \in T, s\in S^{-} \setminus \{1\}\}, \\
  \bE_0 &=& \bigcup_{s \in S^{-}} \{ \{(t_1,s),(t_2,s)\} \mid t_1, t_2 \in T, s(\{t_1,t_2\}) =s\}, \mbox{ and } \\
  \bE_1 &=&    \bigcup_{s \in S^{-} \setminus \{1\}} \{ \{(t_1,s-1),(t_2,s-1)\} \mid t_1, t_2 \in T, s(\{t_1,t_2\}) =s \} \}. 
\end{eqnarray*}
\noindent
Figure~\ref{fig:auxiliary_graph} depicts
  the auxiliary graph corresponding to
  the timetable in Table~\ref{table:timetable}.
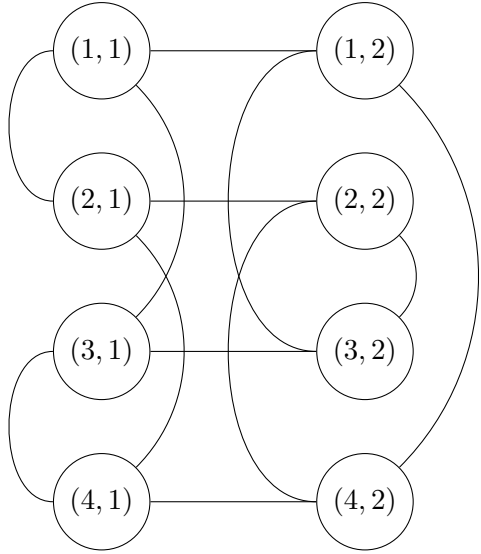
\begin{figure}
\centering
\label{fig:auxiliary_graph}
\begin{tikzpicture}
    \node[draw, circle](N_1_1) at (0,7) {$(1,1)$};
    \node[draw, circle](N_2_1) at (0,5) {$(2,1)$};
    \node[draw, circle](N_3_1) at (0,3) {$(3,1)$};
    \node[draw, circle](N_4_1) at (0,1) {$(4,1)$};
    \node[draw, circle](N_1_2) at (3.5,7) {$(1,2)$};
    \node[draw, circle](N_2_2) at (3.5,5) {$(2,2)$};
    \node[draw, circle](N_3_2) at (3.5,3) {$(3,2)$};
    \node[draw, circle](N_4_2) at (3.5,1) {$(4,2)$};

    \draw (N_1_1) -- (N_1_2);
    \draw (N_2_1) -- (N_2_2);
    \draw (N_3_1) -- (N_3_2);
    \draw (N_4_1) -- (N_4_2);
    \draw (N_1_1) to [out=-45,in=45] (N_3_1);
    \draw (N_2_1) to [out=-45,in=45] (N_4_1);
    \draw (N_1_1) to [out=180,in=180] (N_2_1);
    \draw (N_3_1) to [out=180,in=180] (N_4_1);
    \draw (N_2_2) to [out=-45,in=45] (N_3_2);
    \draw (N_1_2) to [out=-45,in=45] (N_4_2);
    \draw (N_1_2) to [out=180,in=180] (N_3_2);
    \draw (N_2_2) to [out=180,in=180] (N_4_2);
\end{tikzpicture}
\caption{Auxiliary graph}
\end{figure}

The following is our main theorem.

\begin{THEO}
  \label{thm:main}
    Let $\tau$ be a given timetable and 
    $\minB (\tau)$ the optimal value 
    of the BMP
    defined by $\tau.$ 
   Then, $\minB (\tau)$ is equal to
    the size of the minimum OCT 
    of the auxiliary graph $G(\tau).$
\end{THEO}
 
First, we show that 
 an HA-assignment consistent with
 a given timetable $\tau$ 
 gives an OCT-map of auxiliary graph $G(\tau).$
Let $\mathbf{Z}$ be an HA-assignment consistent 
 with $\tau.$
We construct a map
 $\OCTmap_{\mathbf{Z}}^{}: \bV \rightarrow 
   (0, 1, 2)$ by 
\[
\OCTmap_{\mathbf{Z}}^{} (t,s) = 
\left\{
\begin{array}{ll}
  1 & (\mbox{if } z_{t,s}=\Home \mbox{ and } z_{t,s+1}=\Away), \\
  2 & (\mbox{if } z_{t,s}=\Away \mbox{ and } z_{t,s+1}=\Home), \\
  0 & (\mbox{otherwise}).
\end{array}
\right.
\]

\begin{LEMM}
\label{lemma:assignment_to_OCTmap}
The map $\OCTmap_{\mathbf{Z}}$ is an OCT-map,
  where $| \OCTmap_{\mathbf{Z}}^{-1} (0) |$, the size of the OCT, is equal to $\minB(\tau)$.
\end{LEMM}

\begin{proof}
We first show that $\OCTmap_{\mathbf{Z}}^{-1} (\HA)$ is
  a independent set of $G(\tau).$
Let us assume that there exists
  an edge $(v,w)$
  whose endpoints belong to $\OCTmap_{\mathbf{Z}}^{-1} (\HA)$.  
The edge $(v,w)$ belongs to $\bE_H, \bE_0$, or $\bE_1$.
If it belongs to $\bE_H$, then
   there exists a $(t,s) \in \bV$ such that $(v,w) = \{(t,s-1), (t,s)\}$.
Then,   
   $z_{t,s} = \Away$ follows from $\OCTmap_{\mathbf{Z}}(v) = 1$ and
   $z_{t,s} = \Home$ follows from $\OCTmap_{\mathbf{Z}}(w) = 1$.
This is a contradiction.  
If it belongs to $\bE_0$, then
   there exists $(s,t_1,t_2) \in S^- \times T \times T$
   such that $(v,w) = \{(t_1,s), (t_2,s)\}$.
Then,   
   $z_{t_1,s} = \Home$ follows from $\OCTmap_{\mathbf{Z}}(v) = 1$ and
   $z_{t_2,s} = \Home$ follows from $\OCTmap_{\mathbf{Z}}(w)  = 1$.
This is also a contradiction.
The same occurs for the case belonging to $\bE_1$.
Therefore,
  $\OCTmap_{\mathbf{Z}}^{-1} (\HA)$ is a independent set.
Similarly,
  $\OCTmap_{\mathbf{Z}}^{-1} (\AH)$ is a independent set,
  and
  $\OCTmap_{\mathbf{Z}}^{}$ is an OCT-map of $G(\tau)$ from the definition.
The size of the OCT, $| \OCTmap_{\mathbf{Z}}^{-1} (0) |$, is
  equal to $\minB(\tau)$.
\medskip 
\end{proof}
\noindent
Table~\ref{table:OCTmapFromPartial} 
  indicates the OCT-map $\OCTmap_{\mathbf{Z}}^{}$
  which is obtained from the HA-assignment $\mathbf{Z}$
  in Table~\ref{table:partial_ha_assignment}.

\noindent
\begin{table}[h]
    \begin{minipage}{.30\textwidth}
    \centering
    \caption[short]{OCT map $\OCTmap_{\mathbf{Z}}^{}$}
    \label{table:OCTmapFromPartial}
    \begin{tabular}{|l|l|l|}
      \hline Slot & $\mathbf{1}$ & $\mathbf{2}$  \\
      \hline
      \hline team 1 & \HA & \AH   \\
      \hline team 2 & 0   & \AH   \\
      \hline team 3 & 0   & \HA   \\   
      \hline team 4 & \AH & \HA   \\
      \hline
    \end{tabular}
  \end{minipage}
  \begin{minipage}{.30\textwidth}
    \centering
    \caption[short]{OCT map $\OCTmap$}
    \label{table:bad-OCTmap}
    \begin{tabular}{|l|l|l|}
      \hline Slot & $\mathbf{1}$ & $\mathbf{2}$  \\
      \hline
      \hline team 1 & \AH & 0   \\
      \hline team 2 & 0   & 0   \\
      \hline team 3 & 0   & \HA   \\   
      \hline team 4 & \AH & 0   \\
      \hline
    \end{tabular}
  \end{minipage}
  \begin{minipage}{.30\textwidth}
    \centering
    \caption[short]{Partial assignment}
    \label{table:inconsistent-HA-assignment}
    \begin{tabular}{|l|l|l|l|}
    \hline Slot & $\mathbf{1}$ & $\mathbf{2}$ & $\mathbf{3}$  \\
    \hline
    \hline team 1 & \Away & \Home & \UnHA  \\
    \hline team 2 & \UnHA & \UnHA & \UnHA \\
    \hline team 3 & \UnHA & \Home & \Away \\
    \hline team 4 & \Away & \Home & \UnHA \\
    \hline
    \end{tabular}
  \end{minipage}
\end{table}

Next, we discuss a procedure to construct a partial HA-assignment
  from an OCT-map of  $G(\tau).$
Let $\OCTmap:\bV \rightarrow (0,\HA, \AH)$ be an OCT-map
  of $G(\tau)=(\bV, \bE),$ which satisfies that
  both $\OCTmap^{-1}(\HA)$ and $\OCTmap^{-1}(\AH)$
  are independent sets of $G(\tau).$
The above property implies that 
    each  edge $\{(t,s-1), (t,s)\} \in \bE_H$ 
    satisfies that $(\OCTmap (t,s-1), \OCTmap (t,s)) \not \in \{(\HA, \HA), (\AH, \AH)\}$.
Thus, we can uniquely construct 
 a partial HA-assignment $\mathbf{Z}^{\OCTmap}$ defined by   
\[
z^{\OCTmap}_{t,s}=\left\{
    \begin{array}{ll}
        \Home & (\mbox{if } \OCTmap (t,s)=\HA \mbox{ or } \OCTmap (t, s-1)=\AH), \\
        \Away & (\mbox{if } \OCTmap (t,s)=\AH \mbox{ or } \OCTmap (t, s-1)=\HA), \\
        \UnHA & (\mbox{otherwise}).
    \end{array}
    \right.
\]
%
When we apply the above procedure  
 to the OCT-map in Table~\ref{table:OCTmapFromPartial},
 we obtain the partial HA-assignment
 in Table~\ref{table:partial_ha_assignment}. 

A partial HA-assignment obtained 
 from an OCT-map of $G(\tau)$ 
 is not always consistent with
 the given timetable $\tau.$
For example, Table~\ref{table:bad-OCTmap}
  gives an OCT-map $\OCTmap$ of $G(\tau)$ 
  defined by timetable $\tau$ 
  in Table~\ref{table:timetable}.
However, the partial HA-assignment 
  $\mathbf{Z}^{\OCTmap}$,  
  presented in Table~\ref{table:inconsistent-HA-assignment}, 
  is not consistent 
  with $\tau$ in Table~\ref{table:timetable}.
This inconsistency arises
  because the pair of teams 1 and 3 play matches
  at each of their home venues in slot 2,
  which does not satisfy the properties of RRT schedules.  

We will show that
  we can obtain a consistent partial HA-assignment
  by repairing $\alpha$.
To facilitate the description of the proof,
  we introduce some notations.
Let \(t_1\) and \(t_2\) be teams that have a game in slot \(s\),
  that is \(t_2 = \tau(t_1, s)\).
Here,
  we assume $s \geq 2$ and $s \in S^-$.
Then,
  the following relations hold:
\[
 \{(t_1, s), (t_2, s)\} \in \bE_0,  \{(t_1, s-1), (t_2, s-1)\} \in \bE_1,
\]
\[
 \{(t_1, s-1), (t_1, s)\}, \{(t_2, s-1), (t_2, s)\} \in \bE_H.
\]

Thus,
  the sequence 
  $((t_1,s-1),(t_1,s),(t_2,s),(t_2,s-1))$
  forms a 4-cycle on $G(\tau)$.
We call such a 4-cycle
  a {\it rectangular cycle}.
We denote the rectangular cycle
  corresponding to a pair of teams $(t_1, t_2)$ as
  $C(t_1, t_2)$.
A map $\OCTmap$ is {\it inconsistent} on a rectangular cycle $C(t_1, t_2)$
  if the following holds:
\begin{align*}
    & \left( \begin{array}{cc}
    \OCTmap (t_1,s-1) & \OCTmap (t_1,s) \\ 
    \OCTmap (t_2,s-1) & \OCTmap (t_2,s) 
    \end{array} \right)
    \in
    \left\{
    \left( \begin{array}{cc}
    \AH & 0 \\ 
      0 & \HA 
    \end{array} \right) 
    \mbox{, }
    \left( \begin{array}{cc}
    \HA & 0 \\ 
      0 & \AH 
    \end{array} \right)
    \mbox{,}
    \left( \begin{array}{cc}
    0   & \HA \\ 
    \AH & 0
    \end{array} \right)
        \mbox{, }
    \left( \begin{array}{cc}
    0 & \AH \\ 
    \HA & 0 
    \end{array} \right)
    \right\};
\end{align*}
otherwise, it is {\it consistent} on $C(t_1, t_2)$.

A map $\OCTmap$ is {\it locally bipartite} on a rectangular cycle $C(t_1, t_2)$
  if and only if
  $\OCTmap^{-1}(1) \cap C(t_1, t_2)$ and
  $\OCTmap^{-1}(2) \cap C(t_1, t_2)$
  are independents sets on $C(t_1, t_2)$.
  
\begin{LEMM} \label{lem:mapconsistent}
    Let $\tau$ be a given timetable 
      and $G(\tau)=(\bV, \bE)$ be 
      the corresponding auxiliary graph.
    Let $\OCTmap: \bV \rightarrow (0,\HA, \AH)$ be an OCT-map of $G(\tau)$.
    Assume that
      $\alpha$ is consistent on
      every rectangular cycle of $G(\tau)$.
    Then,
      the corresponding partial HA-assignment
      $\mathbf{Z}^{\OCTmap}$ is
      consistent with $\tau.$
\end{LEMM}
\noindent
\begin{proof}
Assume that $\mathbf{Z}^{\OCTmap}$ is inconsistent with $\tau$.
Then, 
  there is an inconsistent pair $(z^{\alpha}_{t,s}, z^{\alpha}_{\tau(t,s),s})$.
If $s \geq 2$, 
  then $\OCTmap$ becomes inconsistent on rectangular cycle $C(t, \tau(t,s))$,
  which contradicts the assumption of the lemma.
If $s=1$,
  then $z^{\alpha}_{t,1} = z^{\alpha}_{\tau(t,1),1} \neq *$ holds,
  which implies that $\OCTmap (t,1)=\OCTmap (\tau (t,1),1) \neq 0.$
It contradicts the assumption that $\OCTmap^{-1}(\HA)$ and $\OCTmap^{-1}(\AH)$ are independent sets.
\end{proof}

\bigskip
The following lemma gives 
 a procedure to construct 
 a consistent partial HA-assignment
 from a given OCT-map.
 
\begin{LEMM} \label{lem:mapconsistent}
  Let $\tau$ be a given timetable 
    and $G(\tau)=(\bV, \bE)$ be 
    the corresponding auxiliary graph.
  If we have an OCT-map 
    $\OCTmap: \bV \rightarrow (0,\HA, \AH)$ of $G(\tau),$
    then there exists an OCT-map $\tilde{\OCTmap}: \bV \rightarrow (0,\HA, \AH)$
      of $G(\tau)$ satisfying that 
    (1)~$|\OCTmap^{-1}(0)|=|\tilde{\OCTmap}^{-1}(0)|$ 
      and 
    (2)~the partial HA-assignment
      $\mathbf{Z}^{\tilde{\OCTmap}}$
      is consistent with $\tau.$
\end{LEMM}

\noindent
\begin{proof}

Assume that 
 the partial HA-assignment $\mathbf{Z}^{\alpha}$
 is not consistent with the given timetable $\tau.$ 
Let $s^*$ be the earliest slot of inconsistent rectangular cycles;
\[ \hspace{-0.5cm}
 s^*=\min \left\{
  s \in \Sminus 
  \left|
    \begin{array}{l}
        \exists (t_1, t_2) \in T \times T,  s = s(t_1, t_2) \\
        \OCTmap \mbox{ is inconsistent on } C(t_1, t_2) 
    \end{array}
  \right. \right\}.
\]
We note that $s^* \geq 2$ holds.
We denote $(t^*_1, t^*_2)$ by a pair of teams
  where $ s^* = s(t^*_1, t^*_2)$ and $\OCTmap$ is inconsistent on $C(t^*_1, t^*_2)$.
From definition,   
\begin{equation*}\label{eq:inconsistentV}
   ( \OCTmap (t^*_1,s^*-1), \OCTmap (t^*_2, s^*) )
  \in \{ (\HA, \AH), (\AH, \HA) \}
  \;\;\; \mbox{(see Table~\ref{table:OCTmap})}.
\end{equation*}
\noindent
\begin{table}[h]
    \centering
    \caption[short]{Inconsistent vertex $(t^*_1,s^*)$}
    \label{table:OCTmap}
    \begin{tabular}{|l|l|l|}
      \hline Slot & $s^*-1$ & $s^*$  \\
      \hline 
      \hline team $t^*_1$ & $\OCTmap (t^*_1,s^*-1)\neq 0$ & $0$   \\
      \hline team $t^*_2$ & $0$  & $\OCTmap(t^*_2,s^*) \neq 0$   \\
      \hline
    \end{tabular}
\end{table}
We denote  $\tau (t^*_2, s^*-1)$ by $t^*_3.$
In this proof,
  we assume that $( \OCTmap (t^*_1,s^*-1), \OCTmap (t^*_2, s^*) ) = (1,2)$
  without loss of generality.
  
\noindent
If \mbox{$s^*=2,$} 
 then vertex $(t^*_2,s^*-1)$ has 
 three neighbors 
 $\{ (t^*_1,s^*-1),(t^*_2,s^*) ,(t^*_3,s^*-1)\} \subseteq \bV.$
When  $s^*>2,$ vertex $(t^*_2,s^*-1)$ has 
 four neighbors $\{ (t^*_1,s^*-1),(t^*_2,s^*) ,(t^*_3,s^*-1),(t^*_2,s^*-2) \}$.

\smallskip

Now we introduce a pair of maps 
$\OCTmap_1$ and $\OCTmap_2$
 (see Tables~\ref{table:OCTmap1} and~\ref{table:OCTmap2})
 defined by 

\begin{eqnarray*}
\begin{aligned}
\OCTmap_1 (t,s) &=
    \begin{cases}
    0 & (\text{if } (t,s)=(t^*_1,s^*-1)), \\
    \OCTmap (t^*_1,s^*-1) & (\text{if } (t,s)=(t^*_2,s^*-1) )\text{,} \\
    \OCTmap (t,s) & (\text{otherwise}),
    \end{cases}
\\    
\OCTmap_2 (t,s) &=
    \begin{cases}
    \OCTmap (t^*_2,s^*) & (\text{if } (t,s)=(t^*_2,s^*-1)), \\
    0 & (\text{if } (t,s)=(t^*_2,s^*)), \\
    \OCTmap (t,s) & (\text{otherwise}).
    \end{cases}
\end{aligned}
\end{eqnarray*}

\noindent
\begin{table}[h]
      \centering
    \caption[short]{Map $\OCTmap_1$}
    \label{table:OCTmap1}
    \begin{tabular}{|l|l|l|}
      \hline Slot & $s^*-1$ & $s^*$  \\
      \hline 
      \hline team $t^*_1$ & $0$ & $0$ \\ 
      \hline team $t^*_2$ & $\OCTmap (t^*_1,s^*-1)\neq 0$  & $\OCTmap(t^*_2,s^*) \neq 0$   \\
      \hline
    \end{tabular}
\end{table}
\begin{table}[h]
    \centering
    \caption[short]{Map $\OCTmap_2$}
    \label{table:OCTmap2}
    \begin{tabular}{|l|l|l|}
    \hline Slot & $s^*-1$ & $s^*$  \\
    \hline 
    \hline team $t^*_1$ & $\OCTmap (t^*_1,s^*-1) \neq 0$ & $0$
     \phantom{$=\OCTmap (t^*_1,s^*)$}    \\
    \hline team $t^*_2$ & $\OCTmap (t^*_2,s^*) \neq 0$  & $0$  \\
    \hline
    \end{tabular}
\end{table}

\begin{figure}
\begin{center}
\begin{tikzpicture}[scale=3.0]
    \draw[thick] (0,2) -- (1,2) -- (1,3) -- (0,3) -- cycle;
    \draw[thick] (1,1) -- (2,1) -- (2,2) -- (1,2) -- cycle;
    \draw[thick] (2,0) -- (3,0) -- (3,1) -- (2,1) -- cycle;
    \draw[thick] (0,0) -- (1,0) -- (1,1) -- (0,1) -- cycle;
    
    \foreach \x / \y in {0/0, 1/0, 0/2, 1/1, 2/0} 
    {
      \fill[blue] (\x,\y) circle (1pt);
      \fill[blue] (\x,\y+1) circle (1pt);
      \fill[blue] (\x+1,\y) circle (1pt);
      \fill[blue] (\x+1,\y+1) circle (1pt);
    }
    \node[draw, circle, minimum size=0.5cm, fill=white, scale=0.55] (A) at (1,1) {$(t^*_2, s^*-1)$};
    \node[draw, circle, minimum size=0.5cm, fill=white, scale=0.55] (A) at (1,2) {$(t^*_1, s^*-1)$};
    \node[draw, circle, minimum size=0.5cm, fill=white, scale=0.75] (A) at (2,1) {$(t^*_2, s^*)$};
    \node[draw, circle, minimum size=0.5cm, fill=white, scale=0.75] (A) at (2,2) {$(t^*_1, s^*)$};

    \draw[->, thick] (0.28,2.85) arc[start angle=120,end angle=-210,radius=0.43]
    node[below, xshift=12.0mm, yshift=-4.0mm, scale=0.8] {$C(t_1^*, \tau(t_1^*, s^*-1))$};

    \draw[->, thick] (1.3, 1.85) arc[start angle=120,end angle=-210,radius=0.43]
    node[below, xshift=10.5mm, yshift=-2.5mm] {$C(t^*_1,t^*_2)$};

    \draw[->, thick] (0.28, 0.85) arc[start angle=120,end angle=-210,radius=0.43]
    node[below, xshift=11.0mm, yshift=-3.0mm, scale=0.8] {$C(t_2^*, \tau(t_2^*, s^*-1))$};

    \draw[->, thick] (2.3, 0.85) arc[start angle=120,end angle=-210,radius=0.43]
    node[below, xshift=10.5mm, yshift=-3.0mm, scale=0.8] {$C(t_2^*, \tau(t_2^*, s^*+1))$};

    \node at (-0.15, 1.07) [draw, align=center] {u};
    \node at (1.15, 0.07) [draw, align=center] {v};
    \node at (-0.15, 0.07) [draw, align=center] {w};
\end{tikzpicture}
\end{center}
\label{fig:4cycle}
\caption{Rectangular cycles connected to $C(t_1^*, t_2^*)$.}
\end{figure}
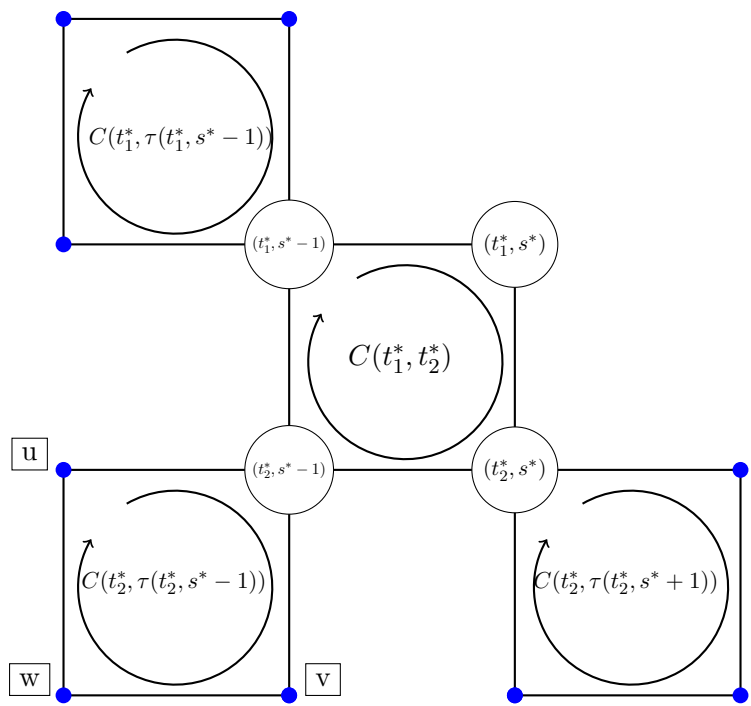
\noindent
Figure~\ref{fig:4cycle} depicts
  three rectangular cycles connected to $C(t_1^*, t_2^*)$.
Both $\OCTmap_1$ or $\OCTmap_2$ are
  locally bipartite on $C(t_1^*, t_2^*)$, $C(t_1^*, \tau(t_1^*, s^*-1))$ and $C(t_2^*, \tau(t_2^*, s^*+1))$.
Both maps are also consistent on them.  
For $C(t_2^*, \tau(t_2^*, s^*-1))$,
  we denote $\{\alpha(t_2^*, s^*-2), \alpha(t_3^*, s^*-1), \alpha(t_3^*, s^*-2)\}$ as $\{(u, v, w)\}$.
Figure~\ref{fig:4cycle} exhibits
  the nodes corresponding to $\{(u, v, w)\}$.

To obtain the desired OCT-map $\tilde{\OCTmap}$,
  we update $\OCTmap$ to $\OCTmap_{\bullet}$
  so that $\OCTmap_{\bullet}$ satisfies that
    (a) $\OCTmap_{\bullet}$ is an OCT-map,
    (b) the number of rectangular cycles where $\OCTmap_{\bullet}$ is inconsistent decreases by 1, and
    (c) $|\OCTmap^{-1}(0)|=|\OCTmap_{\bullet}^{-1}(0)|$ holds.  
For $s^* = 2$,
  both $\OCTmap_1$ and $\OCTmap_2$ satisfy
  those properties.
For $s^* > 2$,
  we have to choose either $\OCTmap_1$ or $\OCTmap_2$ carefully.
The decision on
  which map to choose from follows
\begin{eqnarray*}
  \OCTmap_{\bullet} =
  \begin{cases}
   \alpha_{3-\myalpha{u}} & \text{(if } \myalpha{u} \neq 0 \text{),}\\
   \alpha_{3-\myalpha{v}} & \text{(if } \myalpha{v} \neq 0 \text{),} \\
   \alpha_{\myalpha{w}} & \text{(if } \myalpha{w} \neq 0 \text{),}\\
   \alpha_1 & \text{otherwise.}
 \end{cases}
\end{eqnarray*}

\noindent
First, we demonstrate that the choice is well-defined.
It suffices to show it for the following four cases:

\noindent
(a) $\underline{\myalpha{u}, \myalpha{v} \neq 0, \myalpha{w} = 0}$;
  given the minimality of $s^*$,
  it follows that $\myalpha{u} = \myalpha{v}$.
Therefore, $\alpha_{3-\myalpha{u}} = \alpha_{3-\myalpha{v}}$ holds.
\\
(b) $\underline{\myalpha{u}, \myalpha{w} \neq 0, \myalpha{v} = 0}$;
  as $\OCTmap$ is an OCT-map, $\myalpha{u} \neq \myalpha{w}$ holds.
  Therefore, $\alpha_{3-\myalpha{u}} = \alpha_{\myalpha{w}}$ holds.
\\
(c) $\underline{\myalpha{v}, \myalpha{w} \neq 0, \myalpha{u} = 0}$;
  the proof is the same as that for case (b).
\\
(d) $\underline{\myalpha{u}, \myalpha{v}, \myalpha{w} \neq 0}$;
  as $\OCTmap$ is a OCT-map,
  $\myalpha{u} \neq \myalpha{w}$ 
  and $\myalpha{u} = \myalpha{v}$ holds.
Therefore, 
  $\alpha_{3-\myalpha{u}} = \alpha_{3-\myalpha{v}} = \alpha_{\myalpha{w}}$ holds.

\noindent  
Therefore, the choice is well-defined.
Next, we will show that
  the chosen $\OCTmap_{\bullet}$ is locally bipartite on $C(t_2^*, \tau(t_2^*, s-1))$.
  This follows from
    the conditions
      $\OCTmap_{\bullet}(t_2^*, s-1) \neq \OCTmap_{\bullet}(t_2^*, s-2)$
        and
      $\OCTmap_{\bullet}(t_2^*, s-1) \neq \OCTmap_{\bullet}(t_3^*, s-1)$.
As $\OCTmap_{\bullet}$ is
  locally bipartite on every rectangular cycle,
  it is an OCT-map.
      
Finally, we will show that
  the chosen $\OCTmap_{\bullet}$ is consistent on $C(t_2^*, \tau(t_2^*, s-1))$.
If it is inconsistent,
  then
  $\OCTmap_{\bullet} (t_2^*, s-1)\neq 0, \OCTmap_{\bullet} (t_3^*, s-2) \neq 0$,
  $\OCTmap_{\bullet} (t_3^*, s-1) \neq \OCTmap_{\bullet} (t_3^*, s-2)$
  must hold.
Meanwhile, $\OCTmap_{\bullet} (t_3^*, s-2) = \OCTmap_{\bullet} (t_3^*, s-2)$ must hold if $\OCTmap_{\bullet} (t_3^*, s-2) = w \neq 0$ holds.
This is a contradiction.
Therefore, $\OCTmap_{\bullet}$ is consistent on $C(t_2^*, \tau(t_2^*, s-1))$.

As a result, $\OCTmap_{\bullet}$ satisfies that
  (a) $\OCTmap_{\bullet}$ is an OCT-map,
  (b) the number of rectangular cycles where $\OCTmap_{\bullet}$ is inconsistent decreases by 1, and
  (c) $|\OCTmap^{-1}(0)|=|\OCTmap_{\bullet}^{-1}(0)|$ holds.  
Considering the procedure to
  update $\OCTmap$ by $\OCTmap_{\bullet}$ repeatedly,
  it is bounded above by
    the number of rectangular cycles where $\OCTmap$ is inconsistent.
After this finite procedure,
  we obtain
  an OCT-map $\tilde{\OCTmap}$
  that is consistent on every rectangular cycle.
The corresponding partial HA-assignment $\mathbf{Z}^{\tilde{\OCTmap}}$
  becomes consistent with $\tau$
  from Lemma~\ref{lem:mapconsistent}.
This proves the lemma.
\end{proof}

\medskip
The above proof of the lemma gives a procedure 
  to construct a desired OCT-map  $\tilde{\OCTmap}$
  satisfying that 
  (1)$|\OCTmap^{-1}(0)|=|\tilde{\OCTmap}^{-1}(0)|$ 
   and 
   (2) the partial HA-assignment 
  $\mathbf{Z}^{\tilde{\OCTmap}}$ 
  is consistent with the given timetable $\tau.$ 
We note that it is done in polynomial time.  
Now, we are ready to prove the main theorem.

\noindent
\begin{proof}[Proof of Theorem~\ref{thm:main}]
Let $Z^*$ be an optimal HA-assignment 
  of the BMP defined by $\tau.$
From Lemma~\ref{lemma:assignment_to_OCTmap}, OCT-map $\OCTmap_{Z^*}$ of $G(\tau)$ satisfies that 
  the size of the corresponding OCT 
  $|\OCTmap^{-1}_{Z^*}(0)|$ is equal to $\minB (\tau).$
Thus, the size of the minimum OCT
  of $G(\tau)$, denoted by $\minO (\tau),$
  satisfies that $\minO(\tau)\leq \minB (\tau).$

Let $\OCTmap_*$ be an OCT-map
  corresponding to the minimum OCT
  of $G(\tau).$
Lemma~\ref{lem:mapconsistent} implies that 
  there exists an OCT-map $\tilde{\OCTmap}$
  of $G(\tau)$ that satisfies 
  $|\OCTmap_*^{-1} (0)|=|\hat{\OCTmap}^{-1} (0)|$
  and the partial HA-assignment $Z^{\hat{\OCTmap}}$
  is consistent with $\tau.$
Let $\hat{Z}$ be an HA-assignment 
consistent with $\tau$, 
  which is obtained 
  from $Z^{\hat{\OCTmap}}$ by assigning 
  $\Home$ or $\Away$ to $*$-components,
  as described in Section~\ref{subsection:BMP}.
Then, the number of breaks in 
  $\hat{Z}$ is less than or equal to
  $\minO (\tau)$, which implies that
  $\minB (\tau) \leq \minO (\tau).$
\end{proof}

\medskip

Agarwal, Charikar, Makarychev and Makarychev~\cite{agarwal2005log}
  proposed $O( \sqrt{\log v} )$-approximation algorithm
  for the edge bipartization problem
  defined on a graph with $v$ vertices.
The minimum edge bipartization problem
  finds the smallest edge subset $E'$
  such that the graph $G'$ obtained by deleting $E'$ is bipartite.
As the degree of the auxiliary graph $G(\tau)$ is bounded by 4,
    the approximation algorithm of the OCT problem on $G(\tau)$ can be constructed by
    the approximation algorithm of the edge bipartization problem on $G(\tau)$ with the same approximation ratio.
Consequently, we obtain the following result
  by applying it to the BMP.

\begin{CORO}
There is a randomized polynomial-time \\
$O( \sqrt{\log n} )$-approximation algorithm for BMP.
\end{CORO}

\section{Example}
In this section,
  we present an example of a timetable with $n=4$.
Table~\ref{table:timetable_example}
  indicates an example of a timetable.
We obtain the optimal OCT-map $\OCTmap$
  by solving the OCT problem of the auxiliary graph $G(\tau)$.  
Following the integer linear programming formulation described in~\cite{huffner2009algorithm},
  we solved it
  using the Nuorium Optimizer version 26.1.0~\cite{nuopt},
  an integer linear programming solver.  
  
\begin{table}[h]
  \begin{minipage}{.45\textwidth}
    \centering
    
    \caption{Timetable of $n=4$ \label{table:timetable_example}}
    \begin{tabular}{|c|c|c|c|c|c|c|c|}
    \hline
    Slot & 1 & 2 & 3 & 4 & 5 & 6 & 7 \\ \hline \hline
    1 & 4 & 8 & 5 & 3 & 2 & 6 & 7 \\ \hline
    2 & 3 & 7 & 4 & 8 & 1 & 5 & 6 \\ \hline
    3 & 2 & 6 & 8 & 1 & 7 & 4 & 5 \\ \hline
    4 & 1 & 5 & 2 & 7 & 6 & 3 & 8 \\ \hline
    5 & 7 & 4 & 1 & 6 & 8 & 2 & 3 \\ \hline
    6 & 8 & 3 & 7 & 5 & 4 & 1 & 2 \\ \hline
    7 & 5 & 2 & 6 & 4 & 3 & 8 & 1 \\ \hline
    8 & 6 & 1 & 3 & 2 & 5 & 7 & 4 \\ \hline
    \end{tabular}
  \end{minipage}
  \begin{minipage}{.45\textwidth}
  \centering
    \label{table:example_OCTmap}
    \caption{Optimal OCT-map $\OCTmap$}
\begin{tabular}{|c|c|c|c|c|c|c|}
\hline
Slot & 1 & 2 & 3 & 4 & 5 & 6 \\ \hline \hline
1 & 2 & 1 & 2 & 1 & 2 & 1 \\ \hline
2 & 2 & 1 & 0 & 2 & 0 & 1 \\ \hline
3 & 1 & 2 & 1 & {\bf 2} & {\bf 0} & 1 \\ \hline
4 & 1 & 2 & 1 & 2 & 1 & 2 \\ \hline
5 & 2 & 0 & 1 & 2 & 1 & 2 \\ \hline
6 & 2 & 1 & 2 & 1 & 0 & 2 \\ \hline
7 & 1 & 2 & 0 & {\bf 0} & {\bf 1} & 2 \\ \hline
8 & 1 & 0 & 2 & 1 & 2 & 1 \\ \hline
\end{tabular}

  \end{minipage}
\end{table}

\noindent
We can anticipate an inconsistent rectangular cycle $C(3,7)$.
It can be recovered
  by the procedure
  described in Lemma~\ref{lem:mapconsistent}.
Table~\ref{table:example_revised_OCTmap} indicates the updated OCT-map
  $\tilde{\OCTmap}$
  and Table~\ref{table:example_HA-assignement} indicates the optimal HA-assignment obtained by the OCT-map.

\begin{table}[h]
  \begin{minipage}{.45\textwidth}
    \centering  
    \caption{Updated OCT-map $\tilde{\OCTmap}$\label{table:example_revised_OCTmap}}
\begin{tabular}{|c|c|c|c|c|c|c|}
\hline
Slot & 1 & 2 & 3 & 4 & 5 & 6 \\ \hline \hline
1 & 2 & 1 & 2 & 1 & 2 & 1 \\ \hline
2 & 2 & 1 & 0 & 2 & 0 & 1 \\ \hline
3 & 1 & 2 & 1 & {\bf 2} & {\bf 0} & 1 \\ \hline
4 & 1 & 2 & 1 & 2 & 1 & 2 \\ \hline
5 & 2 & 0 & 1 & 2 & 1 & 2 \\ \hline
6 & 2 & 1 & 2 & 1 & 0 & 2 \\ \hline
7 & 1 & 2 & 0 & {\bf 1} & {\bf 0} & 2 \\ \hline
8 & 1 & 0 & 2 & 1 & 2 & 1 \\ \hline
\end{tabular}
  
  \end{minipage}
  \begin{minipage}{.45\textwidth}
    \centering
    \caption{Optimal HA-assignment\label{table:example_HA-assignement}}
\begin{tabular}{|c|c|c|c|c|c|c|c|}
\hline
Slot & 1 & 2 & 3 & 4 & 5 & 6 & 7 \\ \hline \hline
1 & A &  H & A & H & A & H & A\\ \hline
2 & A & H & A & A & H & H & A \\ \hline
3 & H & A & H & A & H & H & A \\ \hline
4 & H & A & H & A & H & A & H \\ \hline
5 & A & H & H & A & H & A & H \\ \hline
6 & A & H & A & H & A & A & H \\ \hline
7 & H & A & A & H & A & A & H \\ \hline
8 & H & A & A & H & A & H & A \\ \hline
\end{tabular}
  \end{minipage}
\end{table}

\footnotesize

\bibliographystyle{plain}
\bibliography{ref}

\begin{thebibliography}{10}

\bibitem{agarwal2005log}
Amit Agarwal, Moses Charikar, Konstantin Makarychev, and Yury Makarychev.
\newblock \mbox{$O(\sqrt{\log n})$} approximation algorithms for min uncut, min
  2{CNF} deletion, and directed cut problems.
\newblock In {\em Proceedings of the thirty-seventh annual ACM symposium on
  Theory of computing}, pages 573--581, 2005.

\bibitem{akiba2016branch}
Takuya Akiba and Yoichi Iwata.
\newblock Branch-and-reduce exponential/{FPT} algorithms in practice: A case
  study of vertex cover.
\newblock {\em Theoretical Computer Science}, 609:211--225, 2016.

\bibitem{briskorn2008feasibility}
Dirk Briskorn.
\newblock Feasibility of home--away-pattern sets for round robin tournaments.
\newblock {\em Oper. Res. Lett.}, 36(3):283--284, 2008.

\bibitem{de1981scheduling}
D.~de~Werra.
\newblock Scheduling in sports.
\newblock {\em Studies on Graphs and Discrete Programming}, 11:381--395, 1981.

\bibitem{ELF2003343}
Matthias Elf, Michael J{\"u}nger, and Giovanni Rinaldi.
\newblock Minimizing breaks by maximizing cuts.
\newblock {\em Oper. Res. Lett.}, 31(5):343--349, 2003.

\bibitem{goemans1995improved}
Michel~X Goemans and David~P Williamson.
\newblock Improved approximation algorithms for maximum cut and satisfiability
  problems using semidefinite programming.
\newblock {\em J. ACM}, 42(6):1115--1145, 1995.

\bibitem{henz2001scheduling}
Martin Henz.
\newblock Scheduling a major college basketball conference—revisited.
\newblock {\em Oper. Res.}, 49(1):163--168, 2001.

\bibitem{huffner2009algorithm}
Falk H{\"u}ffner.
\newblock Algorithm engineering for optimal graph bipartization.
\newblock {\em Journal of Graph Algorithms and Applications}, 13(2):77--98,
  2009.

\bibitem{iwata2014linear}
Yoichi Iwata, Keigo Oka, and Yuichi Yoshida.
\newblock Linear-time {FPT} algorithms via network flow.
\newblock In {\em Proceedings of the twenty-fifth annual ACM-SIAM symposium on
  Discrete algorithms}, pages 1749--1761. SIAM, 2014.

\bibitem{jansen2011polynomial}
Bart~MP Jansen and Stefan Kratsch.
\newblock On polynomial kernels for structural parameterizations of odd cycle
  transversal.
\newblock In {\em International Symposium on Parameterized and Exact
  Computation}, pages 132--144. Springer, 2011.

\bibitem{kawarabayashi2010almost}
Ken-ichi Kawarabayashi and Bruce Reed.
\newblock An (almost) linear time algorithm for odd cycles transversal.
\newblock In {\em Proceedings of the twenty-first annual ACM-SIAM symposium on
  Discrete Algorithms}, pages 365--378. SIAM, 2010.

\bibitem{kolay2020faster}
Sudeshna Kolay, Pranabendu Misra, MS~Ramanujan, and Saket Saurabh.
\newblock Faster graph bipartization.
\newblock {\em J. Comput. System Sci.}, 109:45--55, 2020.

\bibitem{lokshtanov2014faster}
Daniel Lokshtanov, NS~Narayanaswamy, Venkatesh Raman, MS~Ramanujan, and Saket
  Saurabh.
\newblock Faster parameterized algorithms using linear programming.
\newblock {\em ACM Trans. Algorithms}, 11(2):1--31, 2014.

\bibitem{lokshtanov2009simpler}
Daniel Lokshtanov, Saket Saurabh, and Somnath Sikdar.
\newblock Simpler parameterized algorithm for {OCT}.
\newblock In {\em IWOCA 2009}, volume 5874 of {\em Lecture Notes in Computer
  Science}, pages 380--384. Springer, 2009.

\bibitem{miyashiro2002characterizing}
Ryuhei Miyashiro, Hideya Iwasaki, and Tomomi Matsui.
\newblock Characterizing feasible pattern sets with a minimum number of breaks.
\newblock In {\em PATAT 2002}, volume 2740 of {\em Lecture Notes in Computer
  Science}, pages 78--99. Springer, 2002.

\bibitem{miyashiro2005polynomial}
Ryuhei Miyashiro and Tomomi Matsui.
\newblock A polynomial-time algorithm to find an equitable home--away
  assignment.
\newblock {\em Oper. Res. Lett.}, 33(3):235--241, 2005.

\bibitem{miyashiro2006semidefinite}
Ryuhei Miyashiro and Tomomi Matsui.
\newblock Semidefinite programming based approaches to the break minimization
  problem.
\newblock {\em Comput. Oper. Res.}, 33(7):1975--1982, 2006.

\bibitem{nemhauser1998scheduling}
George~L Nemhauser and Michael~A Trick.
\newblock Scheduling a major college basketball conference.
\newblock {\em Oper. Res.}, 46(1):1--8, 1998.

\bibitem{nuopt}
{NTT DATA Mathematical Systems Inc.}
\newblock {N}uorium {O}ptimizer.
\newblock \url{https://www.msi.co.jp/solution/nuopt/index.html}.
\newblock Accessed: 2024-09-05.

\bibitem{post2006sports}
Gerhard Post and Gerhard~J Woeginger.
\newblock Sports tournaments, home--away assignments, and the break
  minimization problem.
\newblock {\em Discrete Optim.}, 3(2):165--173, 2006.

\bibitem{Rasmussen_2008}
Rasmus~V. Rasmussen.
\newblock Scheduling a triple round robin tournament for the best \mbox{Danish}
  soccer league.
\newblock {\em European J. Oper. Res.}, 185(2):795--810, 2008.

\bibitem{rasmussen2006timetable}
Rasmus~V Rasmussen and Michael~A Trick.
\newblock The timetable constrained distance minimization problem.
\newblock In {\em CPAIOR 2006}, volume 3990 of {\em Lecture Notes in Computer
  Science}, pages 167--181. Springer, 2006.

\bibitem{reed2004finding}
Bruce Reed, Kaleigh Smith, and Adrian Vetta.
\newblock Finding odd cycle transversals.
\newblock {\em Oper. Res. Lett.}, 32(4):299--301, 2004.

\bibitem{trick2001schedule}
Michael~A Trick.
\newblock A schedule-then-break approach to sports timetabling.
\newblock In {\em PATAT 2000}, volume 2079 of {\em Lecture Notes in Computer
  Science}, pages 242--253. Springer, 2000.

\bibitem{van2020complexity}
David Van~Bulck and Dries Goossens.
\newblock On the complexity of pattern feasibility problems in time-relaxed
  sports timetabling.
\newblock {\em Oper. Res. Lett.}, 48(4):452--459, 2020.

\bibitem{van2022optimizing}
David Van~Bulck and Dries Goossens.
\newblock Optimizing rest times and differences in games played: an iterative
  two-phase approach.
\newblock {\em J. Sched.}, 25(3):261--271, 2022.

\bibitem{zeng2012separation}
Lishun Zeng and Shinji Mizuno.
\newblock On the separation in 2-period double round robin tournaments with
  minimum breaks.
\newblock {\em Comput. Oper. Res.}, 39(7):1692--1700, 2012.

\end{thebibliography}

\end{document}